%% file: paper1.tex
\documentclass[11pt]{article}
\usepackage{arxiv}
\usepackage{times}
\usepackage{microtype}
\usepackage{graphicx}
\usepackage{booktabs}
\usepackage{tabularx}
\usepackage{longtable}
\usepackage{amsmath,amssymb,amsthm}
\usepackage{enumitem}
\usepackage{xcolor}
\usepackage{hyperref}
\usepackage{tikz}
\usetikzlibrary{arrows.meta,positioning,calc,fit,backgrounds}
\usepackage{caption}
\usepackage{float}
\usepackage{fancyhdr}
\usepackage{titlesec}
\usepackage{url}
\usepackage{docmute}

\definecolor{ink}{HTML}{1B1B1B}
\definecolor{accent}{HTML}{8D3B56}
\definecolor{soft}{HTML}{F4F1EE}
\definecolor{blue}{HTML}{315A7D}
\definecolor{graytext}{HTML}{555555}
\hypersetup{
  colorlinks=true,
  linkcolor=ink,
  citecolor=accent,
  urlcolor=accent,
  pdftitle={From Product-Centred Retrieval to Experience-Led Commerce},
  pdfauthor={Nafiul I. Khan, Mansura Habiba, Rafflesia Khan}
}
\setlist{leftmargin=*,topsep=0.3em,itemsep=0.15em}
\titleformat{\section}{\large\bfseries}{\thesection.}{0.5em}{}
\titleformat{\subsection}{\normalsize\bfseries}{\thesubsection.}{0.5em}{}

\newtheorem{proposition}{Proposition}

\theoremstyle{definition}

\title{\textbf{From Product-Centred Retrieval to Experience-Led Commerce:}\\
Twelve Candidate Design Principles for Fashion E-Commerce User Experience}
\author{ 
    Nafiul I. Khan \\
	Individual Researcher \\
	Khulna, Bangladesh \\
	\texttt{earthkhan01@gmail.com} \\
	\AND
     \href{https://orcid.org/0000-0001-9051-1370}{\includegraphics[scale=0.06]{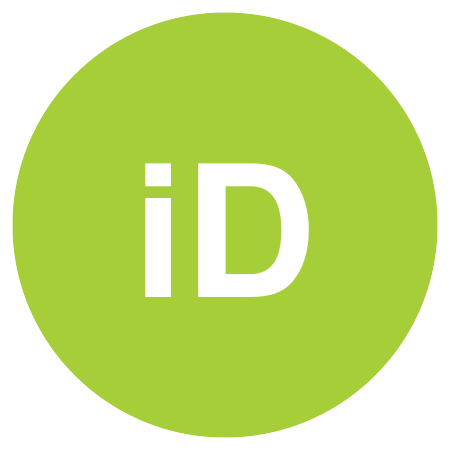}\hspace{1mm}Mansura Habiba} \\
    	Principal Platform Architect - Agentic AI\\
    	IBM Software\\
    	Dublin, Ireland \\
    	\texttt{mansura.habiba@gmail.com} \\
   \And
    Rafflesia Khan\\
	Software Developer\\
	IBM Software\\
	Dublin, Ireland \\
	\texttt{khan.rafflesia@gmail.com} \\
}
\date{}

\begin{document}
\maketitle
\def\maketitle{}
\input{paper1-source.tex}

\end{document}

%% file: paper1-source.tex
\maketitle

\begin{abstract}
Online shopping is now a routine retail channel, but the dominant storefront still assumes that shopping begins with a product query and proceeds through an individual listing to checkout. In fashion e-commerce, demand is increasingly formed elsewhere: through creator and social-media content, short video and represented looks, upcoming situations, owned garments, saved plans, values, and partially authored outfit compositions. Search, recommendation, and generative assistance can reduce retrieval and comparison effort, but they can also collapse hard requirements into inferred preference, amplify commercial priorities, obscure whether a represented item corresponds to the current seller's product, and present a seemingly suitable candidate before stock and delivery feasibility are established. The resulting problem is not only product discovery but \emph{representational discontinuity}: the storefront retains a candidate product while discarding the source, relations, constraints, evidence, and reasons that made it relevant.

This paper proposes twelve candidate Experience-Led Commerce design principles for high-constraint, relational fashion e-commerce, surfaced through design-led induction while building VogueDrop, a multi-vendor prototype. The principles address multi-entry discovery, experience continuity, relational exploration, preference sovereignty, evidence-scoped correspondence, recommendation-time feasibility, customer-compatible commercial ranking, adaptive but stable workspaces, attributable transaction authority, outcome-linked learning, shared composition authoring, and accountable human--agent handoff. The paper uses fashion as an intentionally bounded domain in which fit, composition, material, identity-sensitive preference, seller fragmentation, visual correspondence, and delivery timing make the interaction breakdown observable; it does not claim universal applicability across e-commerce. Each candidate principle is paired with a prespecified hypothesis, primary behavioural outcome, and rejection condition. A formative critical-incident study and a preregistered matched-interface experiment are specified, with user-experience and platform-facing outcomes reported separately. No empirical superiority is claimed before those studies are completed.
\end{abstract}


\section{Introduction}

Online shopping is now a routine retail channel rather than an exceptional alternative. The more consequential shift is in how shopping begins. Demand is increasingly formed across creator posts, social feeds, short video, represented looks, recommendations, saved collections, and personal planning rather than inside the final storefront. Social-commerce research connects social interaction, information exchange, and creator influence to consumer decision processes, while also showing that exposure does not automatically produce purchase \cite{zhang2016social,hu2022social}. In a 2024 United States survey, 62\% of adult TikTok users reported using the service for product reviews or recommendations \cite{faverio2024}. The commercial problem has therefore moved beyond attracting a shopper to a catalogue: the storefront must receive a shopping state that may already contain a narrative, an expectation, a relationship among items, or uncertainty inherited from another platform.

Fashion makes this change especially visible. A shopper may still seek a named product, but may instead try to reproduce a look encountered in creator or entertainment content, adapt a style around an owned garment, prepare for an occasion or journey, satisfy material and value constraints, or author a complete outfit before selecting its products. These are high-constraint, relational tasks. Their success depends not only on the attributes of an individual listing but also on compatibility among items, fit and body evidence, the role of owned products, seller correspondence, budget, stock, return conditions, and delivery timing. Fashion recommendation research addresses compatibility and fit \cite{deldjoo2023,nestler2021}; the interface question is how those relations and constraints remain available across discovery, comparison, and commitment.

Algorithmic mediation sharpens rather than removes this problem. Search, recommendation, and generative assistance can reduce retrieval and comparison effort and can surface candidates that a shopper might not find manually. The same mechanisms can also merge customer-authored requirements, temporary taste, inferred preference, sponsorship, margin, and inventory pressure into one ranking process. Multi-stakeholder recommendation recognises that customer, provider, and platform objectives differ \cite{abdollahpouri2017}; in a consequential shopping task, commercial value cannot legitimately compensate for a failed hard requirement. A visually persuasive candidate may match the desired style while remaining wrong in material, seller, product identity, availability, or delivery feasibility.

The prevailing fashion storefront nevertheless remains organised around the category, query, product, basket, and order. This structure is efficient when the shopper's problem has already been translated into catalogue attributes. It is structurally incomplete when the point of entry is a creator story, a short video, an owned garment, an upcoming journey, or a remembered composition. In those cases, the initial object is not merely a product but a situated representation containing an origin, a reason for interest, relationships among items, explicit or latent constraints, and varying degrees of evidential support.

The loss occurs at the boundary between discovery and transaction. A creator's account of a complete look is reduced to one or more links; a situational requirement becomes a sequence of filters; and a represented garment becomes indistinguishable from visually similar inventory supplied by another merchant. The product page may remain internally accurate while the shopping experience becomes externally incoherent. The shopper must then reopen the source, remember why an item mattered, reconstruct which products belonged together, re-enter requirements, and infer whether the current listing corresponds to what was represented. We call this failure \emph{representational discontinuity}: a transition in which the commerce system preserves the candidate product but discards task-relevant structure that gave the candidate meaning.

This paper therefore does not propose a universal replacement for e-commerce. It uses fashion e-commerce as a bounded research domain because fit, composition, identity-sensitive preference, material constraints, seller fragmentation, visual correspondence, and delivery deadlines expose the limits of catalogue-first interaction particularly clearly. The problem is neither to replace the product grid with conversation nor to attach generative AI to an unchanged funnel. It is to construct an interaction architecture in which heterogeneous beginnings can remain semantically continuous with comparison and commitment, while direct product retrieval remains efficient. The architectural contribution must remain valid whether interpretation is implemented manually, through rules, or by mixed-initiative computation.

We therefore propose \emph{Experience-Led Commerce}, hereafter \emph{our proposed ELC}, instantiated in VogueDrop through twelve candidate user-experience design principles generated through iterative design critique. Our proposed ELC treats social, situational, saved, compositional, and direct-retrieval entries as typed projections of a shared experience state whose provenance, constraints, relations, evidence, feasibility, authority, compositions, delegations, and outcomes remain inspectable. Figure~\ref{fig:shift} contrasts this architecture with the page funnel. The essential change is not an additional navigation option; it is that the fashion-commerce experience is designed around the state from which the customer or an authorised delegate is shopping.

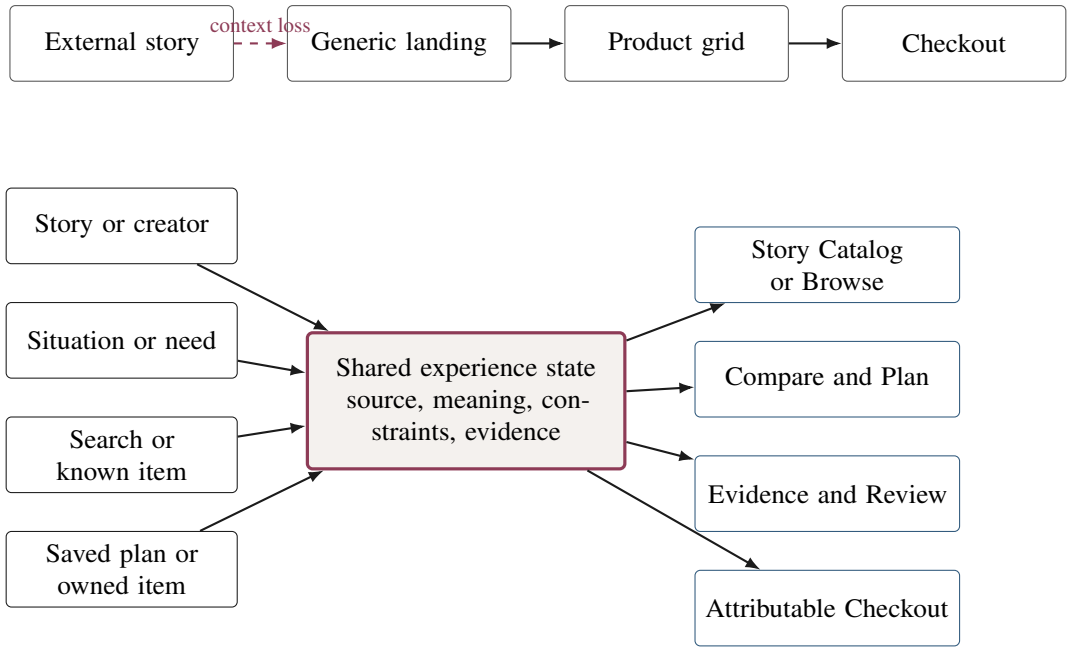
\begin{figure}[H]
\centering
\begin{tikzpicture}[font=\small,node distance=5mm and 7mm]
\tikzset{
old/.style={draw=graytext,rounded corners=2pt,fill=white,minimum width=29mm,minimum height=10mm,text width=27mm,align=center},
entry/.style={draw=ink,rounded corners=2pt,fill=white,minimum width=30mm,minimum height=10mm,text width=28mm,align=center},
hub/.style={draw=accent,very thick,rounded corners=2pt,fill=soft,minimum width=42mm,minimum height=18mm,text width=39mm,align=center},
surface/.style={draw=blue,rounded corners=2pt,fill=white,minimum width=35mm,minimum height=10mm,text width=32mm,align=center},
arr/.style={-{Latex[length=2mm]},thick,draw=ink},
loss/.style={-{Latex[length=2mm]},thick,dashed,draw=accent}
}
\node[old] (social) {External story};
\node[old,right=of social] (landing) {Generic landing};
\node[old,right=of landing] (grid) {Product grid};
\node[old,right=of grid] (checkout) {Checkout};
\draw[loss] (social)--node[above,font=\scriptsize,text=accent]{context loss}(landing);
\draw[arr] (landing)--(grid);
\draw[arr] (grid)--(checkout);
\node[anchor=south west,font=\bfseries] at ([yshift=2mm]social.north west) {Page funnel};

\node[entry,below=14mm of social] (story) {Story or creator};
\node[entry,below=of story] (situation) {Situation or need};
\node[entry,below=of situation] (search) {Search or known item};
\node[entry,below=of search] (saved) {Saved plan or owned item};
\node[hub,right=9mm of situation,yshift=-8mm] (state) {Shared experience state\\source, meaning, constraints, evidence};
\node[surface,right=9mm of state,yshift=18mm] (catalog) {Story Catalog or Browse};
\node[surface,below=of catalog] (compare) {Compare and Plan};
\node[surface,below=of compare] (review) {Evidence and Review};
\node[surface,below=of review] (commit) {Attributable Checkout};
\draw[arr] (story)--(state);
\draw[arr] (situation)--(state);
\draw[arr] (search)--(state);
\draw[arr] (saved)--(state);
\draw[arr] (state)--(catalog);
\draw[arr] (state)--(compare);
\draw[arr] (state)--(review);
\draw[arr] (state)--(commit);
\end{tikzpicture}
\caption{From a page funnel to experience-led commerce. Our proposed ELC preserves entry meaning and provenance while retaining direct product retrieval.}
\label{fig:shift}
\end{figure}

The inquiry is organised by one primary question: \emph{what fashion-commerce interaction architecture can preserve the meaning of heterogeneous shopping entry states without degrading efficient product retrieval?} Three subsidiary questions examine (i) which task-relevant elements are lost at the discovery--commerce boundary, (ii) which representational requirements follow from that loss, and (iii) under which task conditions the resulting architecture improves or harms performance.

The contribution is threefold. First, the paper specifies twelve candidate design principles for fashion e-commerce user experience across discovery, composition, comparison, recommendation, delegated action, transaction, and outcome learning. Second, it makes those principles inspectable in VogueDrop rather than describing them as abstract AI capabilities. Third, it maps every principle to a mechanism, hypothesis, primary outcome, and falsification condition. This structure prevents speed, satisfaction, engagement, creation volume, automation rate, or conversion from standing alone as proof of a better experience.

The remainder of the paper follows the organisation of the broader manuscript. Section~2 positions the contribution against social commerce, retrieval, recommendation, evidence, and HCI research. Section~3 introduces VogueDrop and the change in interaction. Section~4 states the twelve candidate principles, Section~5 supplies the minimum analytical model, and Section~6 defines the evaluation programme. The final sections discuss interpretation, threats to validity, prototype development, ethics, and conclusions.

\section{Related Work and Research Gap}

Our proposed ELC sits at the intersection of research traditions that explain different portions of the shopping journey. This section separates what is already established from the proposed contribution: twelve candidate connected principles for preserving entry meaning, product relations, customer requirements, evidence, feasibility, authority, authored compositions, delegated action, and outcomes across the commerce experience.

\subsection{Social commerce and cross-channel discovery}

Social commerce integrates commercial activity with social interaction, user-generated content, and community influence. Reviews of the field describe social support, trust, information exchange, and social presence as important parts of consumer decision processes \cite{zhang2016social}. Empirical work connects social-media use and perceived informational or social value with purchase intention, while also showing that social exposure does not automatically produce purchase \cite{hu2022social}. In a 2024 United States survey, 62\% of adult TikTok users reported using the service for product reviews or recommendations, compared with 44\% of Instagram users and 37\% of Facebook users \cite{faverio2024}.

These findings establish that discovery occurs outside conventional storefronts. They do not establish how a commerce application should preserve the source narrative, demonstrated composition, creator relationship, or uncertainty after entry. A shoppable story can provide a product link while still discarding the relationship among the story, products actually demonstrated, current substitutes, merchant route, and customer situation.

\subsection{Retrieval, exploration, and continuity}

Information Foraging Theory explains movement through information environments in terms of information scent and expected value \cite{pirolli1999}. Search, filters, categories, and product pages remain efficient when the target is expressible through catalogue attributes. Direct manipulation likewise supports visible objects, reversible actions, and immediate feedback \cite{shneiderman1983}. Our proposed ELC does not replace these mechanisms; it asks how they should coexist with narrative and situational entry, relational planning, evidence inspection, and accountable action.

High-constraint fashion shopping is relational. A candidate may be meaningful only as part of a look, in relation to an owned garment, for a particular occasion, or under a material and deadline requirement. Fashion recommendation research has addressed compatibility, fit, and personalised ranking \cite{deldjoo2023,nestler2021}. The unresolved interface problem is continuity: queries, saved items, product pages, and baskets do not normally preserve why a candidate was retained, which role it serves, or which uncertainty remains.

Multi-stakeholder recommendation additionally recognises that customer, provider, and platform objectives need not coincide \cite{abdollahpouri2017}. The relevant design question is not whether commercial value enters ranking, but whether it is allowed to compensate for a candidate that fails the customer's declared eligibility conditions.

\subsection{Creator narrative and product correspondence}

A representation can be socially persuasive without proving that a current merchant can supply the represented product. Synthetic imagery makes this problem visible, but edited photography, manufacturer imagery, reused media, and stale creator content create the same logical separation. Media origin, demonstrated sample, current listing, seller, and delivered item are different evidence scopes. Studies of generated advertising imagery suggest that disclosure and perceived authenticity affect response \cite{belanche2025,zhang2025}, but an origin label alone cannot establish physical correspondence.

This distinction is particularly important in a story catalogue. Creator credibility is evidence about an experience or narrative; it is not automatically evidence of current stock, material, fit, or seller quality. Our proposed ELC therefore treats story--product correspondence as an inspectable relationship rather than an inferred property of a realistic image.

\subsection{User experience and adaptive human--AI interaction}

ISO~9241-11 defines usability through effectiveness, efficiency, and satisfaction for specified users, goals, and contexts \cite{iso2018}. Task--Technology Fit likewise predicts value when system capability fits the task rather than merely adding functionality \cite{goodhue1995}. Direct manipulation and mixed-initiative research add visible objects, reversible action, feedback, and shared control \cite{shneiderman1983,horvitz1999}; human--AI guidelines further emphasise capability communication, correction, and contextual control \cite{amershi2019}.

These foundations justify neither a universal conversational interface nor a single satisfaction score. Speed, effectiveness, comfort, flexibility, ease, relevance, control, and result quality may move in different directions. VogueDrop therefore treats adaptive assistance as one mechanism within a stable interaction architecture and evaluates the outcome dimensions separately.

\subsection{Agentic commerce infrastructure}

Since this work began, industry has produced protocol-level standards for agent-initiated commerce. Google's Agent Payments Protocol (AP2) uses cryptographically signed mandates to establish authorization, authenticity, and accountability; the workflow analysed by Debi et al. describes an Intent Mandate, Cart Mandate, and Payment Mandate \cite{ap2spec,debi2026whispers}. OpenAI and Stripe's Agentic Commerce Protocol (ACP) standardises agent-initiated checkout around create, retrieve, update, and complete session operations, with cancellation as an additional lifecycle endpoint \cite{acpspec}. Both protocols strengthen what is executed and who authorised it. Neither specifies whether the represented product corresponds to what is delivered, whether a recommendation is feasible at recommendation time, or what evidence and material-change context a person needs to approve a delegated action---the concerns of P5, P6, and P12 respectively. Red-team analysis of AP2 makes the remaining layer explicit: cryptographic guarantees can preserve execution correctness while prompt injection manipulates pre-signature reasoning and mandate construction \cite{debi2026whispers}.

\subsection{Research gap}

The preceding literatures optimise different parts of commerce while assuming different units of analysis. Social-commerce research explains influence and participation; recommender systems rank products; usability research evaluates interaction with a task; and provenance mechanisms qualify individual claims. Their conjunction motivates the following research questions:

\begin{itemize}
\item \textbf{RQ1---Interaction shift.} What task-relevant meaning is lost when social, situational, saved, or relational shopping activity is translated into a product-centred storefront?
\item \textbf{RQ2---Experience quality.} Which changes to the commerce interaction can improve speed, effectiveness, comfort, flexibility, ease, relevance, control, and result quality without degrading efficient known-item retrieval?
\item \textbf{RQ3---Evidence and feasibility.} How should product correspondence, seller evidence, stock, delivery, return, and substitution be represented before a recommendation is presented as suitable?
\item \textbf{RQ4---Preference and commercial authority.} How can customer-authored requirements remain distinct from editable taste, system inference, sponsorship, margin, and inventory objectives?
\item \textbf{RQ5---Delegated interaction.} How can a commerce system expose the same state to humans and authorised agents while preserving evidence, permissions, action boundaries, and timely human intervention?
\item \textbf{RQ6---Validation.} Which observations would support, narrow, or reject each principle of our proposed ELC, and which commercial outcomes require evidence beyond a controlled user study?
\end{itemize}

These questions require an application that makes the inherited interaction and the proposed shift observable. Section~3 therefore reports candidate design needs induced through iterative critique of VogueDrop; Section~4 converts those candidate needs into principles and prespecified developer hypotheses for independent evaluation.

\section{VogueDrop App: Identifying the Interaction Shift}

VogueDrop is used here as a research probe for identifying where the inherited storefront interaction no longer represents the shopping task. Fashion makes the mismatch visible because a customer may begin from a creator story, prepare for an occasion, coordinate with an owned garment, compare fit across bodies, avoid particular materials, or question whether the product represented in an image is supplied by the current seller. A conventional storefront does not necessarily lack the relevant information; it distributes that information across feeds, filters, pages, saved lists, seller records, and customer memory.

The proposed shift is from navigating locally complete pages to manipulating a persistent experience state. Search, categories, product inspection, and checkout remain available, but they are no longer assumed to be sufficient representations of the whole task. Section~3 does not yet state the new principles. It uses VogueDrop to identify which relations must survive the shift if the experience is to become faster, more effective, comfortable, flexible, easy to understand, relevant, controllable, and capable of producing a better result.

\subsection{From page sequence to experience state}

Figure~\ref{fig:surfaces} presents the six primary surfaces. They are not a mandatory funnel. A customer can enter through a story, situational request, saved plan, category, or known-item search; move directly to Browse; compare products in the context of a look; inspect evidence when uncertainty matters; and review the complete transaction before approval.

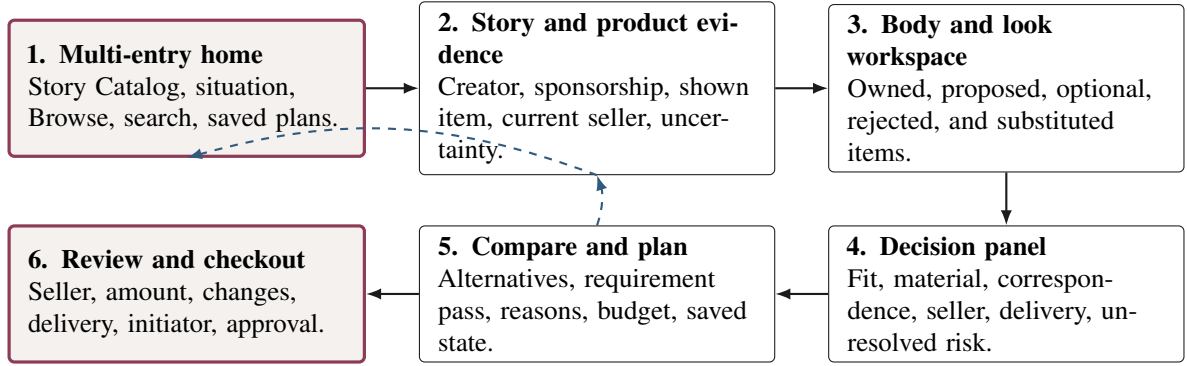
\begin{figure}[H]
\centering
\begin{tikzpicture}[font=\small,node distance=7mm and 7mm]
\tikzset{
box/.style={draw=ink,rounded corners=2pt,fill=white,minimum width=47mm,minimum height=18mm,text width=42mm,align=left,inner sep=4pt},
key/.style={draw=accent,very thick,rounded corners=2pt,fill=soft,minimum width=47mm,minimum height=18mm,text width=42mm,align=left,inner sep=4pt},
arr/.style={-{Latex[length=2mm]},thick,draw=ink},
back/.style={-{Latex[length=2mm]},thick,dashed,draw=blue}
}
\node[key] (entry) {\textbf{1. Multi-entry home}\\Story Catalog, situation, Browse, search, saved plans.};
\node[box,right=of entry] (story) {\textbf{2. Story and product evidence}\\Creator, sponsorship, shown item, current seller, uncertainty.};
\node[box,right=of story] (look) {\textbf{3. Body and look workspace}\\Owned, proposed, optional, rejected, and substituted items.};
\node[box,below=of look] (decision) {\textbf{4. Decision panel}\\Fit, material, correspondence, seller, delivery, unresolved risk.};
\node[box,left=of decision] (compare) {\textbf{5. Compare and plan}\\Alternatives, requirement pass, reasons, budget, saved state.};
\node[key,left=of compare] (checkout) {\textbf{6. Review and checkout}\\Seller, amount, changes, delivery, initiator, approval.};
\draw[arr] (entry)--(story); \draw[arr] (story)--(look); \draw[arr] (look)--(decision);
\draw[arr] (decision)--(compare); \draw[arr] (compare)--(checkout);
\draw[back] (story.south) to[bend right=18] (entry.south);
\draw[back] (compare.north) to[bend right=18] (story.south);
\end{tikzpicture}
\caption{VogueDrop application structure. The surfaces expose different parts of one shopping state while preserving direct Browse and reversible movement.}
\label{fig:surfaces}
\end{figure}

\subsection{Change in user interaction}

Consider a customer preparing for a three-day business trip to Dublin, where rain is expected. She begins with a creator's business-travel story, adds an animal-free requirement, marks Japanese-inspired tailoring as an editable aesthetic preference, and selects an owned blazer as the anchor of the look. The application retains these elements when she opens a dress, compares sizes on the Body Twin, examines similar-body reviews, rejects a visually similar substitute whose material cannot be verified, and adds an alternative from another seller. At review, the system reports which items satisfy the trip, which evidence remains uncertain, whether the complete plan can arrive before departure, and what changed since the plan was saved.

The interaction is neither a long questionnaire nor an autonomous purchasing agent. Requirements can be introduced when they become relevant, and every system proposal remains editable. A style preference is treated as authored taste rather than inferred identity: Japanese-inspired styling has no necessary relation to ethnicity or nationality. An animal-free rule is treated as a hard requirement rather than a weak ranking signal. A generated or creator-supplied image is treated as a representation whose relation to the current seller's product must be established separately.

\subsection{Social Story Catalog and product correspondence}

The Social Story Catalog internalises a discovery pattern that is currently external to most storefronts. Each story retains its creator, sponsorship or affiliate relationship, publication date, demonstrated composition, linked candidates, and correspondence state. The customer may inspect, save, remix, dismiss, or leave for direct Browse. A story view is not counted as progress merely because it increases dwell time; it must advance exploration, comparison, correction, or a deliberately saved state.

Figure~\ref{fig:story-card} shows the central evidence distinction. The story may establish that a creator presented a jacket, but it does not establish that a visually similar listing from a current seller has the same material, construction, or delivered appearance. VogueDrop therefore separates narrative evidence, represented sample, product identifier, seller evidence, substitute status, and unknown correspondence.

\begin{figure}[H]
\centering
\begin{tikzpicture}[font=\small,node distance=6mm and 8mm]
\tikzset{
panel/.style={draw=ink,rounded corners=2pt,fill=white,minimum height=50mm,align=left,inner sep=5pt},
card/.style={draw=blue,rounded corners=2pt,fill=white,minimum width=49mm,minimum height=13mm,text width=44mm,align=left,inner sep=4pt},
alert/.style={draw=accent,very thick,rounded corners=2pt,fill=soft,minimum width=49mm,minimum height=13mm,text width=44mm,align=left,inner sep=4pt}
}
\node[panel,minimum width=55mm,text width=49mm] (story) {\textbf{Creator story}\\[2mm]
\fbox{\rule{0pt}{20mm}\rule{43mm}{0pt}}\\[2mm]
Creator: Aiko Studio\\Relationship: paid collaboration\\Demonstrates: jacket + trousers};
\node[card,right=of story,yshift=17mm] (exact) {\textbf{Demonstrated jacket}\\Exact identifier verified\\Seller A; stock checked};
\node[alert,below=of exact] (similar) {\textbf{Visually similar candidate}\\Not demonstrated\\Material correspondence unknown};
\node[card,below=of similar] (actions) {\textbf{Customer actions}\\Inspect; compare sellers; save; reject substitute; Browse};
\node[draw=accent,dashed,fit=(exact)(similar)(actions),inner sep=4pt,label={[font=\small,text=accent]above:Current commerce candidates}] {};
\end{tikzpicture}
\caption{Story--product evidence in VogueDrop. Persuasive representation and current transaction evidence remain related but non-equivalent.}
\label{fig:story-card}
\end{figure}
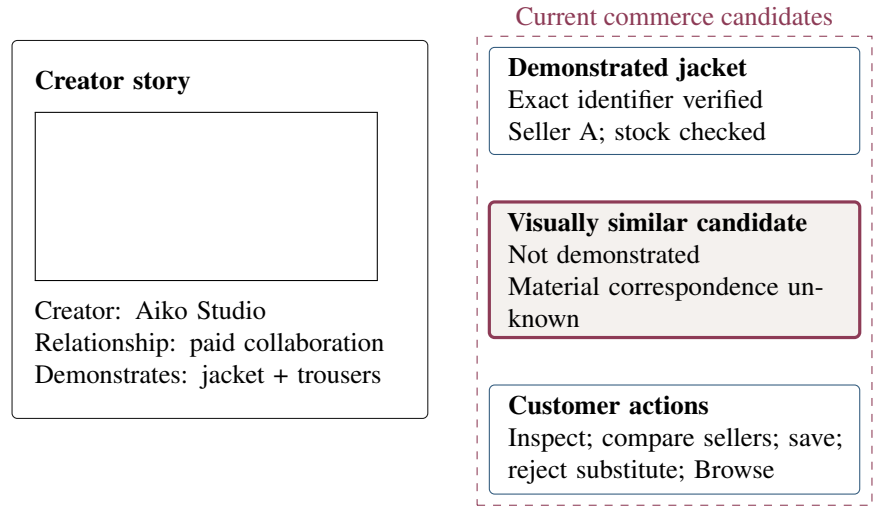

\subsection{From curated looks to a shared Fashion House}

Story-led commerce normally permits a shopper to view, save, or purchase a look assembled by a brand or creator. It rarely gives the shopper a structured means to express why the composition works, adapt it around owned items and requirements, or return the result to a community as an inspectable alternative. VogueDrop therefore includes a \emph{Fashion House}: a composition workspace in which shoppers, creators, and invited collaborators begin from task-relevant layouts---for example, a work-trip capsule, an occasion look, a weekly rotation, or an empty canvas---and populate explicit garment roles rather than merely collecting products in a list.

The interaction borrows the template-and-canvas logic familiar from room-planning applications while changing its unit of validity. A room layout is constrained primarily by space and compatibility; a fashion composition may also depend on body and fit evidence, owned garments, material requirements, weather, budget, delivery time, seller correspondence, and social context. Figure~\ref{fig:fashion-house} therefore places composition between template selection and evidence checking. A shared look is not represented as valid merely because it is visually attractive or popular.

\begin{figure}[H]
\centering
\begin{tikzpicture}[font=\small,node distance=7mm and 4mm]
\tikzset{
fh/.style={draw=ink,rounded corners=2pt,fill=white,minimum width=34mm,minimum height=18mm,text width=30mm,align=center,inner sep=3pt},
fhkey/.style={draw=accent,very thick,rounded corners=2pt,fill=soft,minimum width=34mm,minimum height=18mm,text width=30mm,align=center,inner sep=3pt},
fharr/.style={-{Latex[length=2mm]},thick,draw=ink},
fhback/.style={-{Latex[length=2mm]},thick,dashed,draw=blue}
}
\node[fh] (template) {Choose layout\\capsule, occasion, rotation, blank};
\node[fhkey,right=of template] (compose) {Compose look\\owned, anchor, complement, alternative};
\node[fh,right=of compose] (check) {Check viability\\fit, values, evidence, seller, delivery};
\node[fh,right=of check] (share) {Share with control\\private, invited, public, creator};
\node[fhkey,below=of check] (remix) {Inspect and remix\\attribution and lineage retained};
\draw[fharr] (template)--(compose);
\draw[fharr] (compose)--(check);
\draw[fharr] (check)--(share);
\draw[fharr] (share)--(remix);
\draw[fhback] (remix)--(compose);
\end{tikzpicture}
\caption{The VogueDrop Fashion House. Templates accelerate composition, but constraint and evidence checks precede controlled sharing; subsequent remixing retains authorship and derivation.}
\label{fig:fashion-house}
\end{figure}
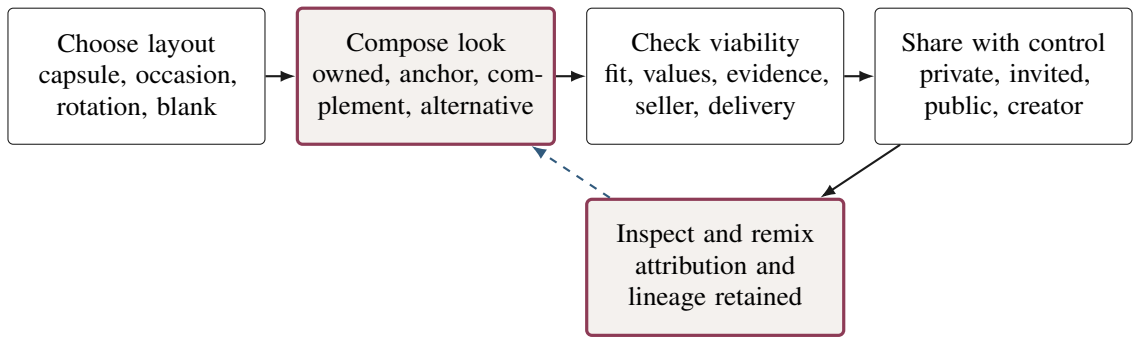

Sharing is governed rather than assumed. The author selects the audience and whether viewing, commenting, copying, or remixing is permitted. Each derived look retains its source composition, changed items, commercial relationships, and unresolved evidence. Community response can support discovery and critique, but likes and reposts cannot override a hard requirement or establish product correspondence. The resulting interaction reveals an additional need: commerce must support customer-authored composition without turning creative participation into untraceable promotion. That need completes the problem statement summarised next.

\subsection{From a human-only interface to delegated interaction}

A second interaction boundary appears when an authorised software agent acts for the customer. Conventional commerce exposes information primarily through pages optimised for human perception. An agent instead requires typed product, seller, evidence, feasibility, preference, composition, and transaction state; declared action preconditions; and machine-interpretable results. Requiring the agent to reconstruct those objects from presentation markup reproduces representational discontinuity at a new boundary. Conversely, a machine-efficient endpoint is insufficient when the customer cannot see what the agent inferred, changed, or proposes to commit.

VogueDrop therefore addresses the human interface and the agent interface to the same canonical experience state. A protocol such as the Model Context Protocol (MCP) can expose data as resources and bounded commerce operations as tools, but protocol exposure is only an implementation mechanism \cite{mcp2025}. Since MCP demonstrates only a resource/tool exposure mechanism, subsequent work has begun to formalise the authority question directly: matching an agent's granted scope to the semantic content of its task \cite{elhelou2025}, and enforcing programmable privilege boundaries at the point of tool invocation \cite{shi2025progent}. The design requirement is stronger: every delegated call carries the acting principal, authority scope, relevant evidence version, preconditions, material differences, and an attributable receipt. Search, comparison, constraint checking, seller verification, and plan preparation may proceed within an explicitly delegated scope. Checkout approval, acceptance of materially changed seller or delivery terms, disclosure of protected preferences, publication of a composition, and other operations outside that scope suspend and request human intervention.

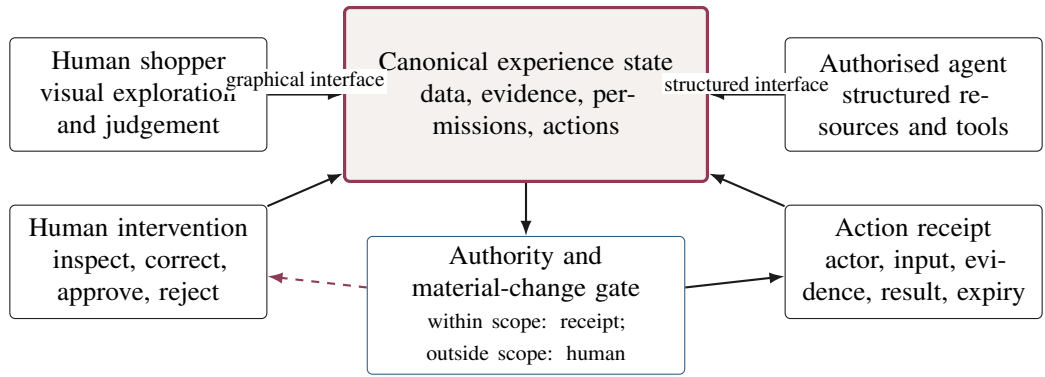
\begin{figure}[H]
\centering
\begin{tikzpicture}[font=\small,node distance=7mm and 10mm]
\tikzset{
actor/.style={draw=ink,rounded corners=2pt,fill=white,minimum width=34mm,minimum height=15mm,text width=30mm,align=center},
state/.style={draw=accent,very thick,rounded corners=2pt,fill=soft,minimum width=48mm,minimum height=23mm,text width=43mm,align=center},
gate/.style={draw=blue,rounded corners=2pt,fill=white,minimum width=42mm,minimum height=18mm,text width=38mm,align=center},
darr/.style={-{Latex[length=2mm]},thick,draw=ink},
escalate/.style={-{Latex[length=2mm]},thick,dashed,draw=accent}
}
\node[actor] (human) {Human shopper\\visual exploration and judgement};
\node[state,right=of human] (shared) {Canonical experience state\\data, evidence, permissions, actions};
\node[actor,right=of shared] (agent) {Authorised agent\\structured resources and tools};
\node[gate,below=of shared] (gate) {Authority and material-change gate\\{\scriptsize within scope: receipt; outside scope: human}};
\node[actor,below=of human] (approval) {Human intervention\\inspect, correct, approve, reject};
\node[actor,below=of agent] (receipt) {Action receipt\\actor, input, evidence, result, expiry};
\draw[darr] (human)--node[above,font=\scriptsize,fill=white,inner sep=1pt]{graphical interface}(shared);
\draw[darr] (agent)--node[above,font=\scriptsize,fill=white,inner sep=1pt]{structured interface}(shared);
\draw[darr] (shared)--(gate);
\draw[escalate] (gate)--(approval);
\draw[darr] (gate)--(receipt);
\draw[darr] (approval)--(shared);
\draw[darr] (receipt)--(shared);
\end{tikzpicture}
\caption{Dual-addressable commerce interaction. Humans and authorised agents operate on the same state; authority and material-change gates determine whether an action executes or returns to the human with its context intact.}
\label{fig:delegated-interaction}
\end{figure}

The intervention must resume the task rather than restart it. The customer receives the agent's proposed action, rationale, supporting and missing evidence, alternatives considered, material changes, and requested authority in the visual interface. Approval or correction then returns as structured state to the agent. This reciprocal handoff reveals a twelfth candidate need: the application must support efficient delegated action without permitting the agent to become an unattributed customer, an unconstrained purchaser, or a separate interface whose state cannot be audited. The resulting candidate set of needs is stated next.

\subsection{Candidate design needs revealed through design-led induction}

The walkthrough reports twelve candidate needs identified through design-led induction while building and iteratively critiquing VogueDrop. They are stated as deficiencies in the inherited interaction so that Section~4 can make the design rationale inspectable rather than present the corresponding principles without an evidence path. These needs were surfaced through iterative design critique of VogueDrop and are treated as candidate needs pending Study~A confirmation, not as an independently validated taxonomy.

\footnotesize
\begin{longtable}{p{0.055\textwidth}p{0.32\textwidth}p{0.54\textwidth}}
\caption{Candidate design needs surfaced through the VogueDrop interaction shift.}\label{tab:needs}\\
\toprule
\textbf{N} & \textbf{Need} & \textbf{Failure in the inherited interaction}\\
\midrule
\endfirsthead
\toprule
\textbf{N} & \textbf{Need} & \textbf{Failure in the inherited interaction}\\
\midrule
\endhead
1 & Preserve heterogeneous beginnings & Story, situation, saved state, and direct search collapse into the same generic landing or grid.\\
2 & Preserve task state & Requirements, reasons, rejected options, and uncertainty disappear across pages or sessions.\\
3 & Preserve product relations & Products are displayed as isolated listings although suitability depends on a look and owned items.\\
4 & Preserve preference authority & Hard values, temporary taste, and inferred preference enter one undifferentiated ranking profile.\\
5 & Distinguish representation from product & A persuasive image or story can be attached to a non-corresponding seller product.\\
6 & Test feasibility before suitability & Relevance is presented before stock, delivery, package, return, and substitution conditions are known.\\
7 & Bound commercial influence & Sponsorship, margin, or inventory pressure can outrank an unmet customer requirement.\\
8 & Match representation to uncertainty & A fixed page or chat stream must serve fit, composition, seller, evidence, and plan comparison equally.\\
9 & Preserve transaction authority & Recommendation, initiation, changed terms, approval, and execution can become difficult to distinguish.\\
10 & Attribute outcomes narrowly & Delivery, use, return, or dispute can update a broad profile rather than the responsible product--seller--decision relation.\\
11 & Enable governed composition authorship & Customers can consume or save complete looks but cannot efficiently construct, test, share, and traceably remix constraint-aware compositions.\\
12 & Coordinate human and delegated interaction & Human-facing pages and agent-facing data paths can expose different state, obscure authority, or lose context when material judgement returns to the customer.\\
\bottomrule
\end{longtable}
\normalsize

These twelve candidate needs delimit the present VogueDrop problem space rather than claiming a complete or canonical taxonomy. The next section connects each candidate need to prior research, states a developer hypothesis, and formulates the candidate design principle intended to satisfy it. Study~A provides the independent mechanism for confirming, revising, merging, or rejecting this design-led decomposition.

\section{Experience-Led Commerce Design Principles}

VogueDrop is organised by twelve candidate design principles that challenge assumptions inherited from catalogue-first commerce. The principles are not justified by novelty or by the availability of agentic AI. Each identifies a user-experience failure, proposes a mechanism, and implies an observable comparison with the conventional storefront. Their common objective is improvement across speed, effectiveness, comfort, flexibility, ease, relevance, control, and result quality while preserving the platform's ability to attract qualified traffic, support useful engagement, and convert viable demand. The separation of outcomes follows task--technology fit and usability research, which requires evaluation relative to users, goals, and contexts rather than feature quantity \cite{goodhue1995,iso2018}.

\subsection{User-experience outcome model}

For participant \(i\), task \(k\), and interface condition \(j\), user experience is represented as the vector
\[
\mathbf{U}_{ikj}=\langle S,E,C,F,A,R,K,Q\rangle,
\]
where \(S\) is speed, \(E\) effectiveness, \(C\) comfort, \(F\) flexibility, \(A\) ease of interaction, \(R\) relevance, \(K\) perceived control, and \(Q\) result quality. These dimensions are reported separately. A faster interaction is not judged better when it produces a constraint violation, and a more enjoyable interaction is not judged effective when the customer fails to identify a misleading substitute.

Platform outcomes are also disaggregated into qualified continuation, useful exploration, viable conversion, return intention, referral intention, and failure indicators such as abandonment, return risk, dispute, and support need. Laboratory measures can test the first-order interaction mechanisms; actual traffic, conversion, return rate, and retention require longitudinal field evidence.

\subsection{Derivation structure}

Each candidate principle follows the same argumentative sequence: research identifies an established interaction or decision property; the VogueDrop walkthrough exposes where the inherited storefront fails to preserve that property; a developer hypothesis states the expected user-experience effect; and the principle prescribes the change to be tested. The hypotheses below are provisional until the formative study is complete, but they make the logic of derivation explicit.

\subsection{P1: Multi-entry discovery}

\textbf{Research and need.} Social-commerce research establishes that product discovery occurs through social interaction and creator content, while Information Foraging Theory explains why entry cues must retain information scent \cite{zhang2016social,hu2022social,pirolli1999}. N1 shows that a generic landing removes those cues. \textbf{Developer hypothesis H1.} Matching story, situation, saved-plan, category, and known-item entrances to the task will improve entry--task fit and qualified continuation; known-item time must remain non-inferior. \textbf{Principle.} Our proposed ELC treats these entry modes as first-class and keeps direct Browse continuously available. The predicted UX effects concern relevance, flexibility, ease, and speed rather than the number of entrances displayed.

\subsection{P2: Experience continuity}

\textbf{Research and need.} Direct manipulation requires visible state and reversible action, yet N2 shows that reasons, requirements, rejected candidates, and uncertainty disappear across commerce surfaces \cite{shneiderman1983}. \textbf{Developer hypothesis H2.} Preserving the minimum task-relevant state will reduce reconstruction events and workload; correction effort must not exceed the reconstruction avoided. \textbf{Principle.} Our proposed ELC carries source, requirements, candidates, evidence, and unresolved questions across surfaces and sessions with visible correction, scope, and expiry. The predicted UX effects are speed, ease, control, and effective resumption.

\subsection{P3: Relational exploration}

\textbf{Research and need.} Fashion recommendation has established that compatibility and fit are relational properties, but N3 shows that the storefront continues to expose isolated listings \cite{deldjoo2023,nestler2021}. \textbf{Developer hypothesis H3.} Representing owned, anchor, complement, optional, alternative, and substitute roles will increase viable-look completion and owned-item reuse; redundant acquisition must not increase. \textbf{Principle.} Our proposed ELC represents candidates through their role in the desired outcome. The predicted UX effects are effectiveness, result quality, relevance, and comparison flexibility.

\subsection{P4: Preference sovereignty}

\textbf{Research and need.} Human--AI guidelines require correction and contextual control, whereas N4 shows that hard requirements, temporary tastes, and inference can enter one ranking profile \cite{amershi2019}. \textbf{Developer hypothesis H4.} Typing preference authority will reduce hard-constraint violations while preserving alternative exploration; maintenance effort must remain within a preregistered boundary. \textbf{Principle.} Our proposed ELC separates customer-authored requirements, editable taste, temporary context, system proposals, and prohibited inference bases. The predicted UX effects are relevance, control, flexibility, comfort, and result quality.

\subsection{P5: Evidence-scoped correspondence}

\textbf{Research and need.} Work on generated advertising imagery shows that origin disclosure and authenticity affect response, but N5 concerns a different relation: whether the current seller product corresponds to what was represented \cite{belanche2025,zhang2025}. \textbf{Developer hypothesis H5.} Separating narrative, represented sample, product identifier, seller, substitute, and unknown states will improve mismatch detection and confidence calibration. \textbf{Principle.} Our proposed ELC exposes evidence by source, scope, and time rather than treating visual realism as product proof. The predicted UX effects are effective judgement, control, comfort, and result quality.

\subsection{P6: Recommendation-time feasibility}

\textbf{Research and need.} Task--Technology Fit requires the information used by the system to match the conditions of the task \cite{goodhue1995}. N6 shows that a relevant fashion item can still fail because stock, seller, delivery, package, return, or substitution conditions are unresolved. \textbf{Developer hypothesis H6.} Evaluating those conditions before presenting suitability will reduce infeasible selections and checkout surprise; stale feasibility displays count as failures. \textbf{Principle.} Our proposed ELC makes current feasibility part of recommendation rather than a downstream checkout discovery. The predicted UX effects are effectiveness, speed to a viable result, comfort, and result quality.

\subsection{P7: Customer-compatible commercial ranking}

\textbf{Research and need.} Multi-stakeholder recommendation recognises distinct customer, provider, and platform objectives \cite{abdollahpouri2017}. N7 shows that combining them in one score can allow commercial value to compensate for customer failure. \textbf{Developer hypothesis H7.} Ranking sponsorship, margin, novelty, and inventory only among customer-eligible candidates will reduce commercially induced requirement violations while preserving viable conversion. \textbf{Principle.} Our proposed ELC separates eligibility from commercial ordering. The predicted UX effects are relevance, control, and result quality; the platform effect is viable rather than raw conversion.

\subsection{P8: Adaptive but stable workspace}

\textbf{Research and need.} Direct manipulation and mixed initiative support task-appropriate views when objects, feedback, and correction remain stable \cite{shneiderman1983,horvitz1999}. N8 shows that fit, composition, evidence, seller, and plan comparison cannot be represented equally well by one fixed page or chat stream. \textbf{Developer hypothesis H8.} Selecting an allowlisted view for the unresolved task will reduce active actions and completion time; navigation error, recall, accessibility, and control must remain non-inferior. \textbf{Principle.} Our proposed ELC adapts representation while preserving semantics, navigation, approval, fallback, and rollback. The predicted UX effects are ease, speed, flexibility, and control.

\subsection{P9: Attributable transaction authority}

\textbf{Research and need.} Mixed-initiative systems require an intelligible allocation of initiative and correction \cite{horvitz1999,amershi2019}. N9 shows that recommendation, preparation, initiation, changed terms, approval, and execution can otherwise become indistinguishable. \textbf{Developer hypothesis H9.} Exposing the initiating actor and requiring renewed approval after a material seller, product, price, delivery, or principal change will increase valid-approval accuracy; repeated confirmation must not produce habituation beyond a defined boundary. \textbf{Principle.} Our proposed ELC makes transaction authority attributable and state-dependent. The predicted UX effects are control, comfort, effectiveness, and error prevention.

\subsection{P10: Outcome-linked learning}

\textbf{Research and need.} Information-system success research distinguishes use and satisfaction from individual and organisational outcomes, while N10 shows that commerce feedback can be assigned to an over-broad profile \cite{delone2003}. \textbf{Developer hypothesis H10.} Linking delivery, use, return, dispute, and recovery to the responsible product--seller--evidence--decision relation will improve a later related task and responsibility attribution; unrelated preferences must remain unchanged. \textbf{Principle.} Our proposed ELC learns at the narrowest supported scope. The predicted UX effects are future relevance, recovery quality, control, and repeat-use effectiveness.

\subsection{P11: Shared composition authoring}

\textbf{Research and need.} Fashion recommendation establishes that outfit quality depends on relations among items, social-commerce research establishes that user-created representations influence discovery, and direct manipulation establishes the value of visible, editable objects \cite{deldjoo2023,zhang2016social,shneiderman1983}. These findings do not imply that a saved product collection is an authored composition. N11 shows that the inherited storefront lets customers consume brand- or creator-assembled looks but provides little structure for constructing a look, expressing its rationale and constraints, or publishing an inspectable derivation. \textbf{Developer hypothesis H11.} Relative to an unstructured saved list, a task-relevant composition template with editable garment roles, viability checks, audience controls, and remix lineage will reduce the active actions required to produce a viable outfit and increase correct reuse by a recipient. The hypothesis is rejected if templates reduce composition diversity or constraint satisfaction, if authoring effort exceeds the avoided search and coordination effort, or if sharing obscures authorship, privacy, sponsorship, or product uncertainty. \textbf{Principle.} Our proposed ELC provides a Fashion House in which customers and creators can start from layouts, assemble and test looks, choose who may view or alter them, and share traceable compositions whose item-level evidence and feasibility remain inspectable. P3 concerns how the system represents relations among products; P11 concerns who can author, communicate, and legitimately transform those relations. The predicted UX effects are effectiveness, flexibility, ease, control, relevance, and result quality. Platform-facing effects are qualified creation, useful referral, and return participation rather than raw posting volume.

\subsection{P12: Dual-addressable interaction and accountable handoff}

\textbf{Research and need.} Mixed-initiative interaction requires an intelligible allocation of initiative, and human--AI guidance requires capability communication, correction, and control \cite{horvitz1999,amershi2019}. MCP demonstrates a protocol mechanism through which an application can expose structured resources, tools, and elicitation, but it does not by itself determine which commerce actions an agent may perform or how a consequential task returns to a person \cite{mcp2025}. N12 shows that a visual site and a machine-facing service can otherwise become two inconsistent applications. The authority boundary is grounded in least privilege: an agent should receive only the capabilities required for its delegated task \cite{schneider2003}.

The narrower authority and oversight question underlying N12 has also received recent direct attention. Work on decoupling human-approval logic from ad hoc, per-application workflows argues that oversight should be designed as a first-class architectural layer rather than added as a filter after agent output \cite{cheng2026hitl}. Work reframing agent security as an agent--human interaction problem treats policy specification, runtime approval, and scope configuration as human-facing security mechanisms that require their own design and evaluation methods \cite{wang2026ahi}.

\textbf{Developer hypothesis H12.} Relative to screen interpretation followed by manual task reconstruction, a structured agent-facing contract over the canonical experience state will reduce extraction and transfer errors and shorten delegated search, comparison, and plan-preparation tasks. At an authority or material-change boundary, a state-preserving human handoff will increase correct approval and rejection without increasing comprehension time beyond a preregistered margin. The hypothesis is rejected if the agent receives weaker evidence than the human interface, exceeds delegated authority, cannot identify the acting principal, executes a material change without renewed approval, or returns insufficient context for an informed human decision. \textbf{Principle.} Our proposed ELC makes commerce state dual-addressable: perceptually appropriate for the human and structurally explicit for the agent, but semantically identical in products, constraints, evidence, permissions, and transaction status. Agent actions are typed as inspect, propose, prepare, or execute; their availability depends on delegated scope. MCP may implement the resource and tool boundary, while the design principle requires identity, least authority, preconditions, material-difference detection, elicitation, resumable handoff, and action receipts independently of protocol. The predicted UX effects are speed, ease, effectiveness, control, comfort, and result quality. Platform-facing effects are valid delegated completion and lower support burden, not maximum automation.

The twelve candidate derivations address the design-led needs identified in Section~3, but the reasoning pattern is more general than VogueDrop. Section~5 abstracts that pattern into a reusable model for design research and demonstrates its application to both VogueDrop and Fabula.

\section{A Reusable Model for Principle-Driven Design Research}

The reasoning used to derive our proposed ELC can be expressed as a reusable design-research model. The model begins with research and situated critique, identifies an interaction breakdown, formulates the experience need that the inherited design fails to satisfy, and converts that need into a falsifiable developer hypothesis and a design principle. An artefact then operationalises the principle; observation supports, narrows, or rejects the hypothesis and feeds the revision back into the principle. Figure~\ref{fig:derivation-model} makes this evidence chain explicit.

\subsection{Design-principle derivation and validation cycle}

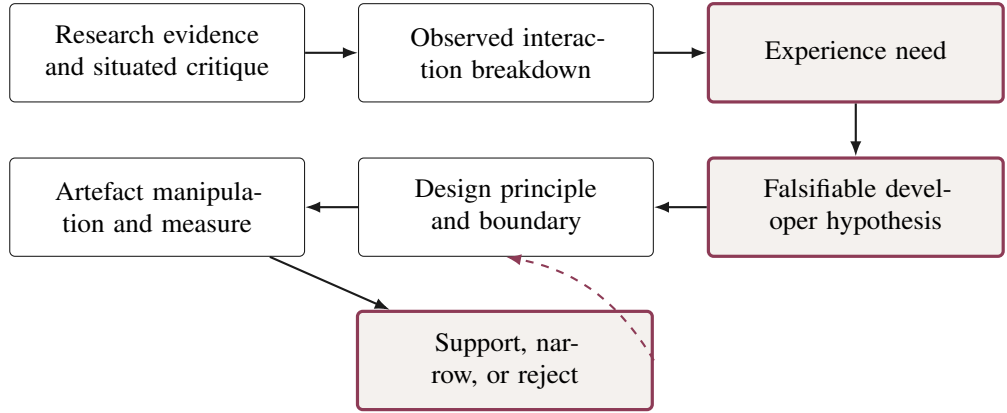
\begin{figure}[H]
\centering
\begin{tikzpicture}[font=\small,node distance=7mm and 7mm]
\tikzset{
stage/.style={draw=ink,rounded corners=2pt,fill=white,minimum width=39mm,minimum height=13mm,text width=35mm,align=center},
key/.style={draw=accent,very thick,rounded corners=2pt,fill=soft,minimum width=39mm,minimum height=13mm,text width=35mm,align=center},
arr/.style={-{Latex[length=2mm]},thick,draw=ink},
rev/.style={-{Latex[length=2mm]},thick,dashed,draw=accent}
}
\node[stage] (evidence) {Research evidence and situated critique};
\node[stage,right=of evidence] (breakdown) {Observed interaction breakdown};
\node[key,right=of breakdown] (need) {Experience need};
\node[key,below=of need] (hyp) {Falsifiable developer hypothesis};
\node[stage,left=of hyp] (principle) {Design principle and boundary};
\node[stage,left=of principle] (artefact) {Artefact manipulation and measure};
\node[key,below=of artefact,xshift=46mm] (decision) {Support, narrow, or reject};
\draw[arr] (evidence)--(breakdown); \draw[arr] (breakdown)--(need); \draw[arr] (need)--(hyp);
\draw[arr] (hyp)--(principle); \draw[arr] (principle)--(artefact); \draw[arr] (artefact)--(decision);
\draw[rev] (decision.east) to[bend right=24] (principle.south);
\end{tikzpicture}
\caption{Reusable principle-driven design-research model. A principle is derived from an evidenced breakdown, operationalised in an artefact, and retained only through a prespecified observation and boundary condition.}
\label{fig:derivation-model}
\end{figure}

The model separates three claims that are often collapsed. Research can establish that a breakdown is plausible or recurrent; it does not establish that a proposed interface resolves it. An artefact can demonstrate that a principle is implementable; it does not establish improvement. Improvement requires a comparison whose primary outcome corresponds to the hypothesised mechanism and whose boundary condition prevents a local gain from being interpreted as general success.

\subsection{Application to VogueDrop}

For VogueDrop, the generic cycle requires a compact analytical representation of what continuity, feasibility, and correspondence mean. The following constructs are limited to those requirements: the state that must persist, the conditions under which a candidate is eligible, and the evidential limit of a product representation.

\subsection{Persistent experience state}

At interaction time \(t\), VogueDrop represents the customer-facing state as
\[
\mathcal{X}_t=\langle s_t,g_t,h_t,p_t,c_t,r_t,e_t,f_t,a_t,o_t\rangle,
\]
where \(s_t\) records discovery source and provenance; \(g_t\) the situation; \(h_t\) hard requirements; \(p_t\) editable preferences and their sources; \(c_t\) candidates; \(r_t\) relations among stories, owned items, and candidates; \(e_t\) evidence and uncertainty; \(f_t\) current stock, seller, delivery, and return feasibility; \(a_t\) action and approval state; and \(o_t\) observed post-purchase outcomes. The tuple is a minimum logical schema, not a requirement to collect every field. Each retained element must have a visible source, purpose, correction path, scope, and expiry.

Let \(Q(z)\) be the information required to complete a task beginning from entry \(z\), and let \(T(z)\) be the state exposed after the interface transition. The lost task information is
\[
L(z,T)=Q(z)\setminus T(z).
\]
A reconstruction event occurs when the customer must reacquire, re-enter, or correct an element of \(L\). P1 and P2 predict that VogueDrop reduces such events in story-led and situational tasks, but not that additional state benefits a known-item task.

\begin{proposition}[Representational discontinuity]
Let \(z_1\) and \(z_2\) be entry states for which the interface exposes the same state, \(T(z_1)=T(z_2)\), although the task-correct actions differ, \(a^{\star}(z_1)\neq a^{\star}(z_2)\). Let \(\pi=\Pr(z_1\mid T)\). Any deterministic or randomised decision rule \(\delta(T)\) that acts without reconstructing a distinguishing element has conditional error bounded by
\[
\Pr\!\left[\delta(T)\neq a^{\star}(z)\mid T\right]\geq \min\{\pi,1-\pi\}.
\]
Reducing error below this bound requires recovering at least one task-relevant element from \(Q(z_1)\triangle Q(z_2)\).
\end{proposition}
\begin{proof}
Because the exposed state is identical, a state-only rule must use the same action distribution for both entries. Always choosing \(a^{\star}(z_1)\) produces error \(1-\pi\), while always choosing \(a^{\star}(z_2)\) produces error \(\pi\). Randomising between them produces a convex combination of those errors and cannot improve on their minimum; any other action is wrong for both entries. Therefore the minimum state-only error is \(\min\{\pi,1-\pi\}\). A lower error requires information that distinguishes the entries, which must be recovered from the task-relevant difference \(Q(z_1)\triangle Q(z_2)\).
\end{proof}

\paragraph{Worked prediction.}
Consider 40 matched transitions into an otherwise identical product card: 20 begin from a creator representation whose exact product--seller correspondence is verified, and 20 begin from a visually similar substitute for which correspondence is unknown. The correct actions differ---continue with the verified candidate versus request evidence or reject the substitute---but a card that omits the source--product relation exposes the same decision state. With \(\pi=0.5\), Proposition~1 gives a state-only error floor of \(0.5\), or at least 20 incorrect actions in 40 trials. With a 70:30 task mixture, the corresponding floor is \(0.3\). These values are illustrative design predictions rather than observed results. Study~B tests whether preserving the distinguishing state reduces correspondence errors below the matched state-only benchmark and whether the reduction is mediated by fewer reconstruction events.

\subsection{Candidate viability}

A candidate \(c\) is eligible for commercial ranking only when it passes the customer's declared requirements and the task's minimum feasibility conditions:
\[
V(c\mid\mathcal{X}_t)=I_h(c)\,I_e(c)\,I_f(c),
\]
where \(I_h\) indicates satisfaction of hard requirements, \(I_e\) indicates sufficient evidence for consequential claims, and \(I_f\) indicates seller and delivery feasibility. Each indicator is binary at the eligibility boundary; uncertainty is not converted into a small positive score. Sponsorship, margin, novelty, and predicted conversion may rank candidates only within \(\{c:V(c\mid\mathcal{X}_t)=1\}\). This ordering operationalises P4--P7 without claiming that commercial objectives are illegitimate.

The eligibility function is also an analysis rule rather than notation used only to describe the interface. For participant \(i\), task \(k\), condition \(j\), and selected candidate \(c_{ikj}\), define
\[
Y^{V}_{ikj}=V(c_{ikj}\mid\mathcal{X}_{t}),
\qquad
Y^{VC}_{ikj}=I_{\mathrm{commit}}\,V(c_{ikj}\mid\mathcal{X}_{t}),
\]
where \(Y^{V}\) is eligible completion and \(Y^{VC}\) is viable conversion or viable plan retention. A commercial-compensation violation is
\[
Y^{C}_{ikj}=\mathbb{1}\!\left[V(c_{ikj}\mid\mathcal{X}_{t})=0
\;\land\; \exists c'\in c_t:V(c'\mid\mathcal{X}_{t})=1\right].
\]
Thus a fast selection with \(V=0\) is recorded as a failure rather than an efficiency gain, and P7 is rejected when commercial ordering selects an ineligible candidate despite an eligible alternative.

\subsection{Representation--product correspondence}

Let \(x\) denote a story image or other product representation, and let \(p\) denote the physical product supplied by a particular seller. The interface may display a correspondence claim only with evidence \(E_{x,p}\) whose scope links the representation to that product and seller. Image-origin disclosure alone concerns how \(x\) was produced; it does not establish \(E_{x,p}\).

\begin{proposition}[Visual non-identifiability]
If the same observed representation \(x\) is compatible with two materially different products \(p_1\) and \(p_2\), no decision rule using only \(x\) can determine with certainty which product will be delivered.
\end{proposition}
\begin{proof}
The rule receives the same input in both possible worlds and must return the same judgement, although the delivered product differs. It is incorrect in at least one world. Seller-, sample-, batch-, or delivery-linked evidence is therefore required.
\end{proof}

Proposition~2 determines the evidence indicator in the viability model. When a representation \(x\) remains compatible with materially different seller products and no scoped evidence \(E_{x,p}\) resolves the relation, \(I_e(c)=0\); visual similarity or image-origin disclosure cannot promote the candidate into the eligible set. The proposition is therefore tested through correspondence classification, false-assurance rate, and the frequency with which participants select a candidate that the frozen task truth marks as evidence-ineligible.

The two theorem-style claims have received preliminary support through exploratory survey and interview work conducted alongside the iterative development of VogueDrop. At this stage, that evidence is treated as an initial verification of practical relevance, not as a confirmatory scientific test of the theorems or their predicted behavioural consequences. Study~A will independently examine whether the underlying breakdowns recur in participant accounts, and Study~B will test their measurable implications under controlled interface conditions.

\subsection{Short mapping to Fabula}

Fabula provides a useful external illustration because its content domain differs while its research logic follows the same cycle \cite{mirowski2026fabula}. Literature on narrative theory and AI writing, together with design interviews and writer sessions, supplies the research evidence. The team then identifies breakdowns and tensions: a chat-dominant interface does not equally support improvisational ``Gardeners'' and plan-oriented ``Architects,'' and hierarchical scene--beat structures can encode culturally situated assumptions.

Those breakdowns become experience needs for alternative workflows, fine-grained author control, culturally contestable structure, and iterative feedback. Fabula formulates developer hypotheses around four design choices, including a custom Gardener/Architect interface and hierarchical narrative planning. The app operationalises those choices through story plans, scene and beat structures, editable snapshots, custom views, and revision controls rather than treating the use of an LLM as the design contribution.

Design interviews, writing sessions, surveys, expert feedback, and wider testing then produce observations that support some mechanisms and qualify others. Some writers valued the custom interface while others still wanted conversational interaction; participants also challenged structural and cultural assumptions. In terms of Figure~\ref{fig:derivation-model}, these are not implementation failures to hide but evidence used to narrow the principle and revise the artefact.

The mapping does not claim that Fabula and VogueDrop share design principles. It shows that the same research model can connect domain evidence, an interaction breakdown, a falsifiable design choice, an artefact, and a revision decision. The following evaluation programme applies that model to the twelve candidate principles of our proposed ELC.

\section{Evaluation Programme}

The evaluation treats the twelve candidate principles as separable design claims with prespecified falsification conditions. This avoids vague statements that additional context, adaptation, delegation, or AI should improve experience. Each principle is supported only when its proposed mechanism changes the designated primary outcome in the predicted direction without violating its boundary condition.

\subsection{Principle-level hypotheses}

\begin{longtable}{p{0.035\textwidth}p{0.32\textwidth}p{0.22\textwidth}p{0.285\textwidth}}
\caption{Hypotheses, primary outcomes, and falsification conditions for the twelve candidate principles.}\label{tab:hypotheses}\\
\toprule
\textbf{P} & \textbf{Falsifiable hypothesis} & \textbf{Primary outcome} & \textbf{Evidence against}\\
\midrule
\endfirsthead
\toprule
\textbf{P} & \textbf{Falsifiable hypothesis} & \textbf{Primary outcome} & \textbf{Evidence against}\\
\midrule
\endhead
1 & Multi-entry discovery increases correct entry selection and qualified continuation in open-ended tasks. & Entry--task match; qualified continuation. & Direct retrieval slows or stories increase attention without progress.\\
2 & Continuity reduces reconstruction of previously available task information. & Reconstruction-event count. & No reduction or correction burden exceeds avoided reconstruction.\\
3 & Relational exploration increases coherent-plan completion and owned-item reuse. & Viable-look completion; redundant-item count. & No completion gain or unnecessary purchasing increases.\\
4 & Preference sovereignty reduces hard-constraint violations while preserving alternative exploration. & Constraint-violation rate. & Violations persist, alternatives collapse, or maintenance burden increases.\\
5 & Evidence-scoped correspondence increases correct mismatch detection and confidence calibration. & Correspondence accuracy. & Detection does not improve or false assurance increases.\\
6 & Recommendation-time feasibility reduces selection of candidates that cannot meet seller and delivery conditions. & Infeasible-selection rate. & No reduction or feasibility displays become stale and misleading.\\
7 & Customer-compatible ranking preserves viable conversion while reducing commercially induced requirement violations. & Viable-conversion rate. & Sponsorship relaxes eligibility or viable conversion materially declines.\\
8 & An adaptive stable workspace reduces action cost in complex tasks without increasing disorientation. & Active actions; navigation errors. & No efficiency gain or accessibility, recall, or control deteriorates.\\
9 & Attributable authority improves detection and rejection of materially changed transactions. & Valid-approval rate. & Change detection does not improve or confirmation habituation increases.\\
10 & Outcome-linked learning improves later recommendation and responsibility attribution. & Subsequent-task accuracy. & Learning is misattributed, overgeneralised, or requires excessive reporting effort.\\
11 & Shared composition authoring reduces the effort of producing a viable outfit and supports correct recipient reuse. & Active actions; viable-composition and reuse rate. & Effort does not fall, diversity or constraint satisfaction declines, or sharing obscures authorship and privacy.\\
12 & Dual-addressable interaction reduces delegated-task transfer errors while preserving informed human authority. & Transfer-error count; valid-approval rate. & Agent and visual state diverge, authority is exceeded, or handoff comprehension deteriorates.\\
\bottomrule
\end{longtable}

The hypotheses are deliberately narrower than a claim that VogueDrop creates a better overall experience. The integrated system is supported only when principle-level effects are accompanied by improvement in the prespecified user-experience vector and when the known-item task remains within a preregistered non-inferiority margin. The comparative study below defines how those observations will be obtained.

The empirical programme contains a formative critical-incident study followed by a controlled interface study. The first independently tests whether the design-induced needs and proposed breakdowns occur in participants' existing shopping journeys, and it may revise the candidate set and interface scenarios. The second estimates the effects of the implemented principles under matched catalogue and task conditions.

\subsection{Study A: cross-channel reconstruction and principle revision}

Study~A uses two analytically separated phases with shoppers, independent creators, and fashion merchants. Its primary purpose is to test whether the twelve candidate needs in Table~\ref{tab:needs}, which were surfaced through VogueDrop's design process, are independently corroborated in participants' own purchase histories.

\textbf{Phase A1---independent critical-incident elicitation.} Participants reconstruct a recent purchase journey before any ELC terminology, VogueDrop concept, prototype, or N1--N12 label is introduced. Interviewers use neutral clarification and sequence prompts only. The Phase~A1 record is closed and versioned before any concept exposure; only this phase can count as independent corroboration of Table~\ref{tab:needs}.

\textbf{Phase A2---concept and boundary walkthrough.} After Phase~A1 is complete, participants inspect matched catalogue-first and VogueDrop concepts, including adversarial cases involving sponsorship, unavailable demonstrated products, substitute sellers, misleading visual similarity, and delegated action. Phase~A2 evaluates boundary conditions, anticipated burden, terminology, implementation risks, and missing mechanisms. Evidence first elicited in Phase~A2 cannot be counted as independent corroboration of a need.

Recruitment seeks variation in shopping frequency, story-led versus direct-search behaviour, marketplace use, creator practice, and merchant scale. Within each stakeholder group, recruitment continues until both conditions hold: (i) two consecutive Phase~A1 interviews yield no new first-cycle codes, and (ii) the retained categories have sufficient variation, boundary conditions, and negative-case coverage for interpretation. These criteria apply after at least 12 shoppers and 4 participants in each provider group; the stopping decision and supporting evidence are recorded in a sampling memo. Recruitment may extend beyond 16 shoppers or 6 participants in a provider group when either condition remains unmet.

The unit of analysis is a meaning-bearing episode: entry, boundary crossing, reconstruction, verification, correction, comparison, or rejection. Phase~A1 is first coded inductively with in-vivo and process codes, without using N1--N12 as first-cycle labels. After the first-cycle codebook stabilises, second-cycle matrices map emergent categories to Table~\ref{tab:needs} and compare entry mode and stakeholder role. Phase~A2 is coded separately for boundary, burden, implementation, and mechanism evidence. At least two analysts independently code a purposive subset; disagreements revise code definitions rather than being resolved only through majority judgement. Negative, contradictory, conventional-storefront-preference, and boundary-rejection cases are reported systematically in the same disposition matrix as corroborating evidence, with case counts, stakeholder context, and implications stated explicitly.

Study~A independently evaluates the design-led need decomposition and may confirm, revise, merge, or remove candidate principles. After Study~A, Table~\ref{tab:needs} will be revised, annotated, or footnoted to report which candidate needs were corroborated in Phase~A1, which were not observed in an elicited incident despite an applicable opportunity, which were explicitly rejected or contradicted, and which previously unanticipated needs emerged. It does not estimate interface superiority.

\subsection{Study B: controlled comparative evaluation}

Study~B uses a counterbalanced within-subject design with three matched conditions:

\begin{enumerate}
\item \textbf{Conventional Funnel}: social or situational prompt opens a standard landing, product grid, product pages, save, cart, and checkout.
\item \textbf{Disclosure-Enhanced Funnel}: the conventional structure adds creator, sponsorship, seller, and image-origin labels but does not preserve shared state.
\item \textbf{VogueDrop}: the six connected surfaces and structured agent interface implement the twelve candidate principles.
\end{enumerate}

Catalogue, product candidates, prices, merchants, task information, response latency, and available final actions are held constant. Interface and task orders are counterbalanced. Sample size will be determined through a preregistered simulation-based power analysis. The provisional smallest effects of practical interest are a 15-percentage-point absolute increase in P5 mismatch-detection accuracy and a 15-percentage-point absolute reduction in P6 infeasible-selection rate---approximately one additional misleading substitute correctly detected, or one fewer infeasible candidate selected, per seven relevant decisions. These thresholds may be revised before preregistration when Study~A participant judgements clarify practical consequence and a separate Study~B-like timed pilot provides baseline rates; every revision and its rationale will be documented before Study~B recruitment. Experimental timing distributions, variance assumptions, and within-participant correlations will be estimated from that timed pilot, not from the qualitative Study~A walkthroughs, with multiplicity control fixed before data collection.

Study~B compares integrated interface configurations. Task-specific contrasts provide mechanism-consistent evidence mapped to individual candidate principles, but they do not fully isolate component causality because several principles change together. A claim that a single component caused an effect requires a later component ablation or factorial study.

\subsection{Tasks}

\begin{table}[H]
\centering
\caption{Controlled task classes and the mechanism each task tests.}
\label{tab:tasks}
\small
\begin{tabularx}{\textwidth}{p{0.22\textwidth}p{0.44\textwidth}X}
\toprule
\textbf{Task} & \textbf{Scenario} & \textbf{Primary construct}\\
\midrule
Known item & Locate and select an explicitly named jacket from a specified seller and size. & Direct-path non-inferiority.\\
Story led & Follow a sponsored creator look into commerce; one linked item is unavailable and one substitute is visually similar but materially different. & Context continuity and correspondence.\\
High constraint & Construct an animal-free, weather-appropriate business-trip look using an owned blazer, a fixed budget, and a delivery deadline. & Reconstruction and viable completion.\\
Shared composition & Adapt a provided work-trip layout around an owned blazer, share it with an invited collaborator, and assess a constraint-preserving remix. & Authoring effort, viable composition, attribution, and reuse.\\
Delegated preparation & Authorise an agent to prepare a trip look; review its evidence and reject or approve after a seller or delivery term changes. & State-transfer accuracy, authority comprehension, and valid approval.\\
\bottomrule
\end{tabularx}
\end{table}

\subsection{Measures}

Principle-level primary outcomes are defined in Table~\ref{tab:hypotheses}. Behavioural instrumentation records entry choice, direct-Browse invocation, reconstruction, repeated search, backtracking, correction, owned-item reuse, alternative exploration, constraint violations, correspondence classification, infeasible selection, commercial-influence comprehension, navigation error, transaction-change detection, approval, recovery, subsequent-task accuracy, template selection, composition edits, sharing scope, remix lineage, delegated calls, authority checks, handoff events, and transfer errors. For each selected candidate, the frozen task record also computes \(I_h\), \(I_e\), \(I_f\), \(V\), eligible completion \(Y^V\), viable conversion or plan retention \(Y^{VC}\), and commercial-compensation violation \(Y^C\). The system records state-changing actions rather than inferring value from a page view, dwell time, post count, or automation rate.

The user-experience vector is measured through task completion time, active action count, viable completion, NASA-TLX workload \cite{hart1988}, and relevant UEQ dimensions including efficiency, perspicuity, and dependability \cite{laugwitz2008}. Task-specific scales assess comfort, flexibility, relevance, control, evidence comprehension, and result quality. In the delegated-preparation task, authority comprehension is measured immediately after the handoff and before approval through a preregistered five-item check covering the acting principal, the delegated action class, one prohibited action, the material seller or delivery change, and whether renewed human approval is required. Each item is scored correct or incorrect for a total of 0--5; the exact wording, frozen answer key, and no-partial-credit rule are specified before data collection. Valid approval is coded against the frozen task truth, and state-transfer accuracy is the proportion of required product, evidence, constraint, seller, and delivery fields correctly identified. Platform-facing laboratory outcomes include qualified continuation, useful exploration, viable conversion, bounce intention, return intention, and referral intention. They are behavioural or attitudinal precursors, not evidence of actual platform traffic or retention.

\subsection{Analysis}

Counts will use Poisson or negative-binomial mixed models according to observed dispersion. Binary accuracy, violation, approval, eligible-completion \(Y^V\), viable-conversion \(Y^{VC}\), and commercial-compensation \(Y^C\) outcomes will use logistic mixed models. Completion time will be analysed on the log scale. Models include interface condition, task class, and their prespecified interaction as fixed effects, with participant as a random effect. Effect sizes and confidence intervals accompany significance tests. Proposition~1 is evaluated by comparing observed decision error in collapsed-state trials with the prespecified state-only lower bound \(\min\{\pi,1-\pi\}\), then testing whether reconstruction events mediate any reduction achieved by the state-preserving condition.

The viability function determines which completions can enter an efficiency comparison. In the known-item task, the named product, seller, and size fix the requirement, evidence, and feasibility indicators, so a correct completion has \(I_h=I_e=I_f=1\) and therefore \(V=1\). Extra experience state is not expected to improve eligibility in this task; its admissible cost is bounded through non-inferiority. Let \(\tau_{ij}\) be completion time for an eligible known-item trial. The estimand is
\[
\Delta_{\mathrm{NI}}=
\mathbb{E}\!\left[\log \tau_{i,\mathrm{ELC}}-\log \tau_{i,\mathrm{Funnel}}
\mid k=\text{known item},\;V=1\right].
\]
VogueDrop is non-inferior when the upper one-sided confidence bound for \(\Delta_{\mathrm{NI}}\) is below \(\log(1+\delta_{\mathrm{NI}})\), where \(\delta_{\mathrm{NI}}\) is the preregistered maximum acceptable proportional delay. Study~A participant judgements inform what delay would be practically consequential; the numerical margin, timing distribution, and variance assumptions will be finalised before the freeze point using a separate timed pilot that reproduces the Study~B instructions and interface conditions. A \(V=0\) selection is a task failure and cannot be credited as a fast completion. Failure to reject a superiority null is not interpreted as non-inferiority.

Confirmatory tests mapped to principle-level outcomes use a preregistered family-wise or false-discovery procedure. They are interpreted as condition-by-task evidence consistent or inconsistent with the corresponding mechanism, not as isolated component effects. Exploratory relationships among speed, comfort, flexibility, relevance, control, and result quality are reported separately.

\subsection{Study flow}

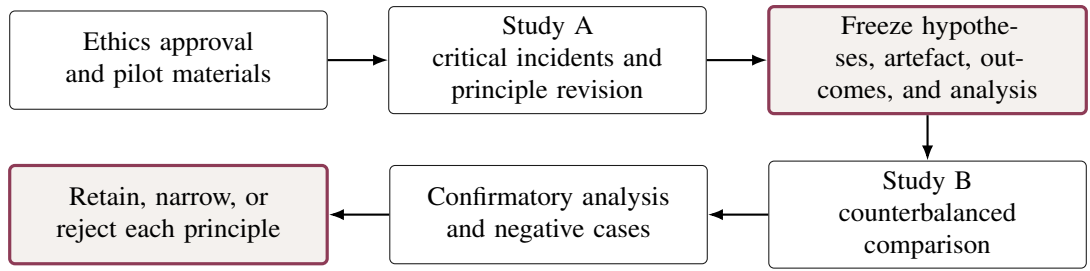
\begin{figure}[H]
\centering
\begin{tikzpicture}[font=\small,node distance=6mm and 8mm]
\tikzset{
stage/.style={draw=ink,rounded corners=2pt,fill=white,minimum width=42mm,minimum height=13mm,text width=38mm,align=center},
key/.style={draw=accent,very thick,rounded corners=2pt,fill=soft,minimum width=42mm,minimum height=13mm,text width=38mm,align=center},
arr/.style={-{Latex[length=2mm]},thick}
}
\node[stage] (ethics) {Ethics approval and pilot materials};
\node[stage,right=of ethics] (a) {Study A\\critical incidents and principle revision};
\node[key,right=of a] (freeze) {Freeze hypotheses, artefact, outcomes, and analysis};
\node[stage,below=of freeze] (b) {Study B\\counterbalanced comparison};
\node[stage,left=of b] (analysis) {Confirmatory analysis and negative cases};
\node[key,left=of analysis] (revision) {Retain, narrow, or reject each principle};
\draw[arr] (ethics)--(a);
\draw[arr] (a)--(freeze);
\draw[arr] (freeze)--(b);
\draw[arr] (b)--(analysis);
\draw[arr] (analysis)--(revision);
\end{tikzpicture}
\caption{Empirical sequence. Formative findings revise the principles before the confirmatory artefact and analysis are frozen.}
\label{fig:study-flow}
\end{figure}

A principle rejected or substantially narrowed by Study~A will not proceed to Study~B confirmatory testing in its original form; Table~\ref{tab:hypotheses}, the task mapping, and the corresponding analysis will be revised before the freeze point. A newly emerged principle enters confirmatory testing only if its mechanism, outcome, boundary condition, and implementation can be fixed before preregistration; otherwise it is reported as an exploratory or future-study result.

The present manuscript specifies the principles, prototype behaviour, hypotheses, and planned analysis; it does not insert anticipated effects into a Results section. Ongoing research is producing versioned prototype conditions, task materials, an event schema, analysis code, and preregistration records. Empirical results remain outside the present manuscript until a corresponding prototype version is frozen and evaluated.

\section{Discussion}

The proposed contribution is a change in e-commerce design principles, not the addition of a stylist agent or a larger set of product-page features. VogueDrop is designed to preserve the conventional strengths of search, Browse, product inspection, and checkout while changing how state, relations, evidence, feasibility, commercial influence, authorship, delegation, authority, and outcomes are represented across them. The twelve candidate hypotheses determine whether that change produces measurable value.

\subsection{Balancing commercial and customer value}

The platform and the shopper do not optimise the same immediate outcome. The shopper may seek a suitable purchase, informed exploration, reuse of an owned item, or a justified decision not to buy. The platform seeks qualified traffic, useful engagement, viable conversion, and repeat participation. This distinction is consistent with perceived-value research and two-sided platform theory, which separate customer benefit and sacrifice from platform coordination and viability \cite{zeithaml1988,rochet2003}. VogueDrop treats the outcomes as compatible only when the interaction produces progress under the customer's declared requirements. Preserved entry meaning may support qualified exploration, but that proposition is tested through reconstruction, continuation, and viable completion rather than asserted from dwell time.

The same reasoning prevents purchase from becoming the sole measure of success. A customer who verifies that a represented item is not supplied by the current seller, reuses an owned garment, or decides that no candidate satisfies an animal-free constraint has completed valuable decision work without purchasing. For the platform, such outcomes may protect confidence and future participation even though they do not maximise the immediate basket. Longitudinal claims about retention or platform growth require field evidence and are not inferred from the laboratory measures proposed here.

\subsection{Preference is authored rather than demographically assigned}

Preference sovereignty is evaluated through observable interface behaviour. The prototype records whether a preference was declared, imported, or proposed; whether it is a hard requirement or editable taste; its scope and expiry; and whether the customer corrected it. A Japanese-inspired style preference is represented only when the customer authors or accepts it; ethnicity, nationality, name, location, and body are excluded as inference bases. An animal-free requirement is encoded as an eligibility condition, while an occasion-specific colour preference may remain negotiable. P4 fails if the system violates these distinctions or if maintaining them costs more effort than the relevance gained.

\subsection{Delegation is not silent automation}

P12 does not state that an agent should reproduce every human interaction through an API, nor that every tool call requires a confirmation dialog. The relevant unit is delegated authority. Reading public catalogue data, comparing eligible candidates, and preparing a reversible plan may be authorised as a bounded class of action. Publishing a look, revealing a protected preference, accepting a changed seller, or committing payment changes the customer's exposure and therefore requires authority that can be inspected at the time of action. The boundary may vary by customer and task, but it cannot be inferred merely from the agent's ability to call a tool.

The human and agent views need not look alike. They must, however, denote the same products, sellers, evidence scopes, constraints, material changes, and transaction state. This semantic equivalence is intended to allow an efficient machine-facing representation and a comprehensible visual representation to coexist. It is also intended to make failure attributable: the study is designed to distinguish incorrect source data, agent interpretation, delegated scope, human approval, and merchant change rather than assigning every error to an undifferentiated automated journey.

This distinction separates two claims that current commerce infrastructure tends to conflate. Cryptographically signed mandates and scoped payment tokens establish non-repudiable authorization for a transaction; they do not establish that the transaction was one the customer would have approved with full information. AP2 red-team analysis makes this explicit: cryptographic enforcement can preserve execution integrity while pre-signature reasoning is manipulated, producing a valid mandate that does not reflect the user's true intent \cite{debi2026whispers}. P12 is positioned at this second layer: not as a competitor to payment-authorization protocols, but as the user-experience condition---evidence, alternatives, material changes, and rationale---under which a signed mandate is informed rather than merely valid.

These claims remain conditional on the implemented interaction, controlled tasks, and evidence available to the study. The next section therefore specifies the validity limits and the staged prototype evaluation required before broader commerce outcomes can be claimed.

\section{Future Works}

The current paper specifies a prospective evaluation and therefore cannot establish effect size, population preference, actual traffic growth, conversion, return reduction, or retention. Controlled tasks cannot reproduce brand attachment, creator-feed competition, real financial consequence, or long-term learning. The twelve candidate principles also interact in the integrated interface; principle-level event measures improve diagnostic power, but causal isolation will require component ablations.

Ongoing research is implementing and versioning the multi-entry home, Social Story Catalog, persistent experience state, Body and Look workspace, Fashion House, Decision Panel, evidence-scoped product comparison, attributable review, structured agent interface, instrumentation, and matched experimental conditions. Pilot work is examining task timing, non-inferiority margins, scale reliability, and power assumptions. Confirmatory work depends on a frozen interface, stimuli, hypotheses, exclusions, and analysis; field evaluation is separately required for actual traffic, viable conversion, return behaviour, and repeat use. These activities extend the same paper and prototype rather than constituting a separate research programme.

Fashion makes fit, composition, visual representation, values, seller fragmentation, and delivery timing visible, but it does not establish transfer to every commerce domain. Generalisation requires replication with domain-specific definitions of viability, evidence, delegated authority, and outcome. These limitations define the evidence still required and prevent the twelve candidate principles from being presented as a complete cross-domain taxonomy.

\section{Conclusion}

The product-centred funnel assumes that the customer's problem has already been translated into something the catalogue can retrieve. Contemporary commerce increasingly violates that assumption. Social narratives, situations, owned items, and saved plans enter the storefront carrying relations and uncertainties that a product identifier cannot preserve. When the transition collapses those distinctions, the customer must reconstruct them or act on an incomplete state.

Our proposed ELC responds by changing the principles that organise discovery, composition, comparison, recommendation, delegation, transaction, and learning, not by placing an agent beside an unchanged product page. VogueDrop makes twelve candidate principles inspectable: multi-entry discovery, continuity, relational exploration, preference sovereignty, evidence-scoped correspondence, feasibility, bounded commercial ranking, stable adaptation, attributable authority, outcome-linked learning, shared composition authoring, and accountable human--agent handoff.

The paper therefore offers twelve candidate design propositions with prespecified falsification conditions rather than claiming that a richer or more automated interface is self-evidently superior. Each proposition survives only when its designated mechanism improves the corresponding outcome and respects its rejection condition. Ongoing prototype and evaluation research is examining whether these principles improve speed, effectiveness, comfort, flexibility, ease, relevance, control, result quality, and viable platform participation.

\section*{Acknowledgements}
This pre-study manuscript was developed through iterative design critique, literature review, and preliminary survey and interview work. That exploratory evidence provides initial support for the two theorem-style claims but is not treated as confirmatory verification. More rigorous scientific evaluation is planned through Study~A and Study~B, and no completed confirmatory participant study is claimed.

\clearpage
\appendix

\section{Study A Interview Guide}

The guide is divided into two phases. Phase~A1 is completed and versioned before any ELC terminology, VogueDrop concept, or targeted principle probe is shown. Creator and merchant prompts in Phase~A2 are asked only of participants with the relevant role.

\subsection{Phase A1: independent critical-incident elicitation}

Interviewers use neutral sequence and clarification prompts only.

\begin{enumerate}
\item Describe a recent fashion purchase journey that began outside the final commerce site. Where did it begin, and what were you trying to achieve?
\item Walk through the journey step by step, including the platforms, people, information, and decisions involved.
\item At each transition, what, if anything, became easier, harder, less clear, or had to be revisited?
\item What worked well enough that you would prefer the existing storefront or process?
\item What important shopping need, breakdown, or form of support did this interview fail to ask about?
\end{enumerate}

\subsection{Phase A2: concept and boundary walkthrough}

Only after Phase~A1 is closed does the interviewer present matched catalogue-first and VogueDrop concepts and ask the following targeted questions.

\begin{enumerate}
\item What did the original story or recommendation establish, and what did it not establish?
\item If the shown item was unavailable, how did you judge whether a substitute was legitimate?
\item When would an in-site Story Catalog help, and when would it become an engagement trap?
\item How quickly must direct search or Browse remain available?
\item Which preferences should persist, expire, or never be inferred?
\item For creators: what control is required over attribution, sponsorship, affiliate disclosure, compensation, reuse, remix, correction, and removal?
\item Describe a time you assembled or shared a complete outfit rather than an individual product. Would a provided layout help, and when would it constrain your expression or encourage unsuitable copying?
\item If another shopper or creator remixed your composition, which authorship, audience, commercial, evidence, and removal information would need to remain attached?
\item For merchants: when does story traffic become qualified demand, and when does it create unsuitable enquiries or misleading expectations?
\item After seeing both concepts, what would make the conventional storefront preferable?
\item At what point should stock, seller, delivery, package split, return route, or substitution affect a recommendation?
\item When is commercial ranking acceptable, and what customer requirement must never be traded for sponsorship, margin, or inventory pressure?
\item Which comparison view would help with the task, and which parts of navigation and control must remain stable?
\item Describe a purchase in which the seller, price, product, delivery, or initiator changed before approval. Which change required renewed consent?
\item Which shopping work could an authorised agent perform for you without interruption, and which data or actions should remain unavailable without a separate decision?
\item When an agent returns a prepared plan for human intervention, what rationale, evidence, uncertainty, alternatives, changes, and action history are required for an informed approval or rejection?
\item After delivery, return, or complaint, what should the platform learn, and which conclusions would be an unjustified generalisation?
\end{enumerate}

\section{Initial Second-Cycle Mapping Codebook}

Phase~A1 first-cycle coding is inductive and does not use this table. The table is applied only after the inductive codebook stabilises, to map emergent categories against Table~\ref{tab:needs}; Phase~A2 may use it for targeted boundary and implementation analysis. ``Need not observed in elicited incident'' is recorded in a participant-by-need disposition matrix after an applicable opportunity is established. It is not treated as participant-reported thematic evidence.

\small
\begin{longtable}{p{0.22\textwidth}p{0.31\textwidth}p{0.39\textwidth}}
\toprule
\textbf{Code} & \textbf{Definition} & \textbf{Illustrative evidence}\\
\midrule
Entry-mode fit & Match between task and story, situation, saved state, Browse, or search entry & Participant prefers search for a named item but stories for open-ended styling.\\
Cross-boundary loss & Meaning discarded when discovery becomes a marketplace visit & Creator, demonstrated look, or occasion disappears on the landing page.\\
Reconstruction & Recovery of prior facts, constraints, relations, or reasons & Participant reopens a story to remember why a product was retained.\\
Story--product correspondence & Relationship between narrative representation and purchasable candidate & Participant notices that a visually similar substitute is not the demonstrated item.\\
Creator attribution and control & Governance of authorship, sponsorship, compensation, reuse, and removal & Creator rejects an unattributed remix or concealed affiliate relationship.\\
Composition authorship & Construction of a look from a template or blank canvas with explicit garment roles, rationale, constraints, and audience & Participant adapts a work-trip capsule and expects an invited recipient to see why each item was included.\\
Remix governance & Conditions governing copying, transformation, attribution, lineage, commercial disclosure, correction, and removal & Participant permits a public copy but requires derived looks to retain authorship and affiliate history.\\
Social evidence boundary & Separation of narrative experience from verified product evidence & Styling demonstration is useful but not treated as material or seller proof.\\
Relational role & Product understood as owned, anchor, complement, optional, or substitute & Participant keeps an owned blazer and rejects an unnecessary duplicate.\\
Preference source & Distinction among declared requirement, editable taste, temporary context, and inference & Participant corrects a style suggestion but preserves an animal-free rule.\\
Engagement trap & Interaction that prolongs attention without progress & Similar promotional stories prevent access to purposeful search.\\
Boundary rejection & Condition under which a principle of our proposed ELC is unnecessary or harmful & Participant prefers the conventional interface for known-item retrieval.\\
Need not observed in elicited incident & No evidence of a candidate need appears in the Phase~A1 incident despite an applicable opportunity; this is non-observation, not explicit rejection & Participant reports no cross-boundary reconstruction in a relevant incident; the disposition matrix records ``not observed'' unless the participant explicitly rejects or contradicts the need.\\
Feasibility timing & Point at which stock, seller, delivery, package, return, or substitution should constrain recommendation & Participant rejects a preferred item after learning that it cannot arrive before the trip.\\
Commercial eligibility boundary & Customer condition that commercial influence must not relax & Participant accepts sponsored ordering only after animal-free eligibility is satisfied.\\
Adaptive-workspace stability & Representation may change while navigation, semantics, accessibility, approval, and rollback remain stable & Participant values a seller comparison but rejects moving controls or altered requirement meaning.\\
Transaction re-authorisation & Material change requiring renewed customer approval & Participant requires approval after seller, price, product, delivery, or purchasing actor changes.\\
Delegation scope & Data and actions that an agent may inspect, propose, prepare, or execute for a specified task and duration & Participant permits comparison and plan preparation but withholds payment execution and protected preference disclosure.\\
Human--agent handoff & State required to resume a delegated task at a human judgement or authority boundary & Participant requires the proposed action, rationale, missing evidence, alternatives, material changes, and action receipt before checkout approval.\\
Interface-state divergence & Semantic difference between the visual interface and structured agent interface & Agent receives stale delivery evidence or a different seller identifier from the page shown to the customer.\\
Scoped outcome learning & Future update limited to the responsible product, seller, evidence, task, or decision state & A poor batch updates seller--batch confidence but not the customer's general aesthetic identity.\\
\bottomrule
\end{longtable}
\normalsize